\documentclass[11pt]{article}

\usepackage[T1]{fontenc}

\usepackage[margin=1in]{geometry}
\usepackage{amsmath,amssymb,amsthm}
\usepackage{mathtools}
\usepackage{xcolor}
\usepackage{enumitem}

\usepackage{tikz}
\usetikzlibrary{calc,decorations.pathreplacing,arrows.meta}

\usepackage[ruled,linesnumbered]{algorithm2e}
\newcommand{\RRmark}[1]{\tikz[remember picture,overlay]{\coordinate (#1);}}
\SetKwFor{RRFor}{for}{do}{end\RRmark{repair-rr-bottom}}

\usepackage{aliascnt}

\usepackage{hyperref}
\usepackage[capitalise]{cleveref}

\newcommand{\eps}{\varepsilon}
\newcommand{\alloc}{A}

\newcommand{\Core}{\textsc{Core}}
\newcommand{\Repair}{\textsc{Repair}}
\newcommand{\Dominate}{\textsc{Dominate}}
\newcommand{\Divide}{\textsc{Divide}}
\newcommand{\Dom}{\mathrm{Dom}}
\newcommand{\Permute}{\textsc{Permute}}
\newcommand{\cat}{\mathbin{+\!\!+}}

\theoremstyle{plain}

\newtheorem{theorem}{Theorem}

\newaliascnt{proposition}{theorem}

\aliascntresetthe{proposition}

\newaliascnt{lemma}{theorem}
\newtheorem{lemma}[lemma]{Lemma}
\aliascntresetthe{lemma}

\newaliascnt{claim}{theorem}
\newtheorem{claim}[claim]{Claim}
\aliascntresetthe{claim}
\crefname{claim}{claim}{claims}
\Crefname{claim}{Claim}{Claims}

\newaliascnt{corollary}{theorem}

\aliascntresetthe{corollary}

\newaliascnt{conjecture}{theorem}

\aliascntresetthe{conjecture}

\theoremstyle{definition}
\newtheorem{definition}{Definition}

\newaliascnt{remark}{theorem}

\aliascntresetthe{remark}

\theoremstyle{plain}

\usepackage{thmtools}
\usepackage{thm-restate}

\title{Cutting a Cake For Human Consumption:\\An Elementary $n$-agent Cake Cutting Protocol}
\author{Alexandros Psomas \\ Purdue University}
\date{}

\begin{document}
\maketitle

\begin{abstract}
We give a bounded envy-free cake-cutting protocol for any number $n$ of agents in the Robertson--Webb model, with a short, self-contained, and elementary correctness proof. The protocol uses $2^{O(n \log n)}$ queries. Most ingredients are already known and well understood in the fair division literature; our main contribution is a simpler organization of these ideas into a protocol whose complete description and proof fit in a few pages.
\end{abstract}

\section{Introduction}\label{sec:intro}

There is an infinitely divisible resource $C$, the ``cake,'' denoted by the interval $[0,1]$. Our goal is to allocate the cake among a set $N$ of $n$ agents. An allocation $\alloc = (\alloc_1, \dots, \alloc_n)$ of the cake to the agents consists of $n$ pieces, where $\alloc_i$ is the piece allocated to agent $i$, such that $\alloc_i \cap \alloc_j = \emptyset$ for all pairs of distinct agents $i,j$. A piece here refers to a finite union of disjoint intervals of $[0,1]$.

Every agent $i \in N$ has a nonnegative integrable valuation density $v_i$, normalized so that $\int_{x=0}^1 v_i(x) dx = 1$. For any piece $P$, her value is $V_i(P) = \int_P v_i(x) dx$. Valuations are: (i) additive: for a set of disjoint pieces $P_1, \dots, P_m$, $V_i (\cup_{j=1}^m P_j) = \sum_{j=1}^m V_i(P_j)$, (ii) non-atomic: for all $0 \leq x \leq y \leq 1, \lambda \in [0,1], \exists z \in [x, y]$ such that $V_i([x, z]) = \lambda \cdot V_i([x,y])$.

Our goal is to produce a complete, envy-free allocation. An allocation $\alloc$ of a cake $C$ is \emph{complete} if $\cup_{i \in N} \alloc_i = C$. In our subroutines, we use partial allocations, which need not be complete. For a partial allocation $\alloc$ of a cake $C$, the \emph{residue} $R$ is the unallocated part, i.e., $R = C \setminus \bigcup_{i \in N} \alloc_i$. A (complete or partial) allocation $\alloc$ is \emph{envy-free} if for all distinct $i,j \in N$, $V_i(\alloc_i) \geq V_i(\alloc_j)$. Envy-freeness compares only the allocated pieces; it imposes no condition on how agents value the residue. We work in the standard Robertson--Webb model~\cite{robertson1998cake,woeginger2007complexity}. This model allows for two types of queries: (1) $\textsc{Eval}_i(I)$: Given an agent $i$ and an interval $I \subseteq [0,1]$, this query returns $V_i(I)$; (2) $\textsc{Cut}_i(x, v)$: Given an agent $i$, a point $x \in [0,1]$, and a value $v \in [0,V_i([x,1])]$, this query returns the smallest point $x' \in [x, 1]$ such that $V_i([x, x']) = v$. 

Readers familiar with cake cutting may proceed directly to~\Cref{sec:conventions}.

\subsection{Related Work}

Cut-and-choose gives an envy-free allocation for $n=2$ agents, and the Selfridge--Conway protocol gives an envy-free allocation for $n=3$ agents~\cite{robertson1998cake}. Brams and Taylor~\cite{brams1995envy} gave an envy-free protocol for any number of agents, but the number of queries it makes is not bounded by any function of $n$ (it depends on the valuations). Whether a \emph{bounded} protocol exists remained open for decades; Procaccia proved an $\Omega(n^2)$ lower bound~\cite{procaccia2009thou}. Aziz and Mackenzie presented bounded protocols for four agents~\cite{aziz2016discretefour} (later improved in~\cite{amanatidis2018improved}) and then for any number of agents~\cite{aziz2016discrete}, with $n^{n^{n^{n^{n^n}}}}$ queries. Dehghani et al.~\cite{dehghani2018envy} adapted the latter to chores.

\paragraph{A personal, and biased, perspective on cake cutting.}

The following discussion combines the published record with personal recollections and my own (necessarily biased) perspective. I will therefore write in the first person.

I am aware of two more works on the $n$-agent problem. Sokolov~\cite{sokolov2023exponential} states, in the abstract of their work, that minor changes to the Aziz--Mackenzie protocol and an improved analysis give a bound of $n^{8n^2(1+o(1))}$ queries; the manuscript is written in Russian (and therefore I have not been able to verify its contents)\footnote{Furthermore, since most ideas in this work were produced with AI assistance, I cannot rule out that some of them originate in this manuscript.} and is, to the best of my knowledge, unpublished. More recently, in another (at the time of this writing) unpublished manuscript, Ye and Bai~\cite{ye2026cutting} give a new protocol with a bound of~$n^{O(1)}2^n$. Their construction uses the \textsc{Core} and \textsc{SubCore} machinery of Aziz and Mackenzie (with one step of \textsc{SubCore} removed). To the best of my knowledge, the only bounded envy-free protocol for $n$ agents that has appeared in a peer-reviewed venue is that of Aziz and Mackenzie~\cite{aziz2016discrete} (FOCS 2016; a research highlight with a simplified \textsc{Core} protocol appeared in CACM~\cite{aziz2020bounded}). I am also aware of a manuscript by Bechler and Bitikofer~\cite{bechler2017counterexample}, who gave a five-agent counterexample to an earlier version of Aziz and Mackenzie's \textsc{SubCore} protocol (as presented in arXiv v11 of~\cite{aziz2016discrete}) and identified an error in its correctness proof. 
The current arXiv version (v12, August 2017)~\cite{aziz2017arxiv} changes \textsc{SubCore} and introduces an imaginary-value construction for tie-breaking; the proof is revised accordingly, now appealing to this construction.\footnote{The counterexample is not mentioned in v12; I record it here as part of the history of the result.} As part of a project for my graduate course ``Algorithmic Economics,'' Gabriel Iskandar, a former undergraduate student at Purdue, also confirmed the error in the earlier version and wrote an independent (five-page) proof of correctness for the revised \textsc{Core}/\textsc{SubCore} subroutine~\cite{gabe2024}.

To be clear, I know of no error in the revised Aziz--Mackenzie protocol, and I am not claiming one. I am, however, claiming that the solution is not widely understood in the fair division community (I do not understand it myself). Given the importance of the problem, the error in the earlier version, and how easy it is now to generate proofs using AI, a proof accessible to independent verification is particularly important. Accessibility also matters for further progress: understanding an existing solution is a natural starting point for improving or simplifying it (and comparing a new result with prior work is expected if one intends to publish it!). In this way, treating the problem as ``solved'' has (as the limited work on envy-free cake cutting in the past decade suggests) shifted the community's attention to other problems, such as the existence of EFX allocations.

This work was prompted in part by an X thread in which Nisarg Shah~\cite{shah2026tweet} announced AI-assisted improvements to several problems in social choice, including the query complexity of envy-free cake cutting for $n$ agents; replies to the thread claimed further improvements. I believe that, for envy-free cake cutting in particular, the fair division community does not need a faster protocol as much as it needs one that can be understood. So I wanted instead to try to use these AI tools to produce a simple protocol whose complete proof I can read, check, and, ideally, teach (in, say, one lecture) in my graduate course. This work is the outcome of this effort.

\paragraph{Use of AI.} All the proofs in this paper were initially generated by large language models (Claude Fable 5.1 by Anthropic and ChatGPT-6 Astra by OpenAI), with guidance from the author. I read and checked every proof line by line, rewrote most of them, and take full responsibility for the correctness of the paper. My main contribution was imparting my own taste: I rejected several proofs, and even faster protocols, as ``too complicated,'' and chose among the alternatives the models produced. I also made the connections to the related literature; as noted above, however, I cannot rule out that an idea suggested by an LLM originates in work I have not read.

\subsection{This work}

We give an \emph{elementary} solution to finding an envy-free allocation for $n$ agents using a bounded number of queries. Our protocol requires $2^{O(n \log n)}$ Robertson--Webb queries; we give a self-contained correctness proof. 

\paragraph{Overview.}

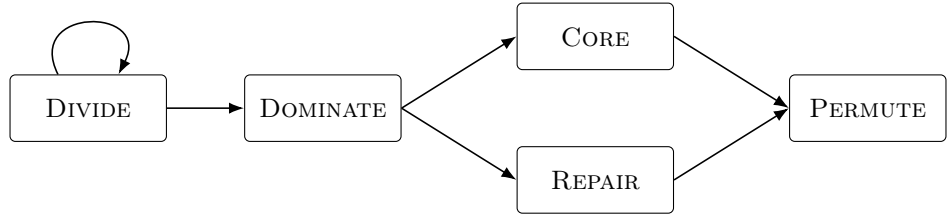
\begin{figure}[th]
\centering
\begin{tikzpicture}[
  box/.style={draw,rounded corners=2pt,minimum height=2.3em,minimum width=5.4em,align=center,font=\small},
  arr/.style={-{Latex[length=2mm]},semithick}]
\node[box] (div)  at (0,0)        {\Divide};
\node[box] (dom)  at (3.1,0)      {\Dominate};
\node[box] (core) at (6.7,0.95)   {\Core};
\node[box] (rep)  at (6.7,-0.95)  {\Repair};
\node[box] (perm) at (10.3,0)     {\Permute};

\draw[arr] (div)  -- (dom);
\draw[arr] (dom.east) -- (core.west);
\draw[arr] (dom.east) -- (rep.west);
\draw[arr] (core.east) -- (perm.west);
\draw[arr] (rep.east)  -- (perm.west);
\draw[arr] ($(div.north)+(-0.4,0)$) .. controls +(-0.45,0.9) and +(0.45,0.9) .. ($(div.north)+(0.4,0)$);
\end{tikzpicture}
\caption{The protocols and their calls.}
\label{fig:calls}
\end{figure}

\Cref{fig:calls} shows the protocols and their calls: \Divide{} calls \Dominate{} and recurses, \Dominate{} calls \Core{} and \Repair{}, and both call \Permute{}. The high-level plan is as follows.

First, we reduce the problem of finding an envy-free allocation to the problem of splitting the agents into two sets, $E$ and $N \setminus E$ (both non-empty), such that every agent in $N \setminus E$ \emph{dominates} every agent in $E$ (see~\Cref{sec: reduce to domination} for the definition); this split allows us to recurse, i.e., the original problem is now reduced to finding an envy-free allocation only among agents in $E$. The concept of domination is not new. Brams and Taylor called it ``irrevocable advantage''~\cite{brams1995envy} (in reference to the protocol of Selfridge and
Conway), and it has been used repeatedly in three-agent~\cite{brams1995envy}, four-agent~\cite{aziz2016discretefour,amanatidis2018improved}, and $n$-agent envy-free protocols~\cite{aziz2016discrete,ye2026cutting} (and even for chores~\cite{dehghani2018envy}). 

\Dominate{}, the protocol that finds a dominated group $E$, maintains an envy-free partial allocation and proceeds in a loop. Each iteration starts by calling \Core{} on the remaining cake. Given a designated cutter, \Core{} returns a partial envy-free allocation in which the cutter receives exactly $1/n$ of her value of its input. If adding this allocation produces a dominated group, we are done. Otherwise, we reassign the newly created pieces. Crucially, there exists a reassignment that introduces only a small amount of envy, while giving some agent $i^*$ an ``advantage'' over a particular bundle. We then call \Repair{}, which restores exact envy-freeness while preserving this advantage. Finally, repeated \Core{} calls, with $i^*$ as cutter, shrink her value of the residue until the advantage becomes domination. This increases the number of agents that $i^*$ dominates, without decreasing any other agent's domination count; we repeat this process until a dominated group is found (each iteration that does not return increases the sum of the agents' domination counts, which is bounded by $n^2$, so the loop must terminate). Creating new dominations by reassigning already allocated pieces and then shrinking the residue with repeated calls to \Core{} is not new either: in different forms, it appears in the four-agent protocols~\cite{aziz2016discretefour,amanatidis2018improved}, in the $n$-agent protocols~\cite{aziz2016discrete,ye2026cutting}, and in the chore protocol of Dehghani et al.~\cite{dehghani2018envy}.

\Core{} asks the designated cutter to divide the input cake into $n$ pieces of equal value, and then calls \Permute{}, a subroutine also used by \Repair{}. Given $n$ disjoint pieces, \Permute{} returns an envy-free partial allocation in which every agent receives a suffix of a different piece and at least one piece is allocated whole. Since the cutter values all $n$ input pieces equally, envy-freeness forces her to receive \emph{exactly} $1/n$ of her value of the input: she cannot receive less than the value of the piece allocated whole, or more than the value of the piece her suffix comes from. This equal-cutting-and-trimming structure also appears in earlier \Core{} protocols~\cite{aziz2016discretefour,aziz2016discrete,amanatidis2018improved}.

\Permute{} is a brute-force search. For every order of the agents and every order of the pieces, it constructs a candidate allocation by processing the agents in that order: each agent starts with her assigned piece whole, and the piece is trimmed just enough that no earlier agent envies it. The protocol returns the first candidate that is envy-free. This search over agent and piece orders is due to Dehghani et al.~\cite{dehghani2018envy}, who introduced it for chores; Aziz and Mackenzie later used it to simplify their \Core{} protocol~\cite{aziz2020bounded}. The main task is to prove that some guess succeeds. For every $\delta>0$, a variant of envy-cycle elimination~\cite{lipton2004approximately} constructs a $\delta$-envy-free allocation containing the candidate of some guess. If every candidate failed, a compactness argument would give a common positive lower bound on the envy of all such allocations; choosing $\delta$ smaller than this bound gives a contradiction.

Finally, \Repair{} receives a partial allocation in which every agent's envy is at most $\eps$ times her value of the residue. It restores exact envy-freeness by adding small pieces of residue to the old bundles and possibly reassigning the resulting bundles (while keeping every old bundle intact). The idea is to put extra cake before each old bundle and call \Permute{}, giving it enough cake to trim without reaching an old bundle. To construct these additions, we partition the residue into (at least $2n^2$ and at most $1/(2\eps)$) cells that every agent values at most $3n\eps$ times her value of the residue, and run round robin for $n-1$ rounds: each agent, in a fixed order, picks her favorite available cell. We place each agent's chosen cells before her old bundle, in the order she picked them, and call \Permute{}. \Permute{} guarantees that the output is envy-free; the main task is to show that no old bundle is cut. We prove this using basic properties of round-robin and a counting argument. Keeping the old bundles intact, together with the small value of the additions, preserves the advantage created by \Dominate{}'s reassignment. Adding residue before old bundles and then trimming while preserving those bundles also appears in the construction of Ye and Bai~\cite{ye2026cutting}.

\paragraph{What is new?} Most of the underlying ideas have precedents, as discussed above. The contributions of this work are (i) their organization into subroutines with simpler interfaces and guarantees; (ii) the round-robin implementation of \Repair{}; and (iii) the direct proof that \Permute{} succeeds on arbitrary input pieces. Together, these give a bounded envy-free protocol with a (relatively) short, self-contained correctness proof.

\subsection{Conventions}\label{sec:conventions}

Here, we state the conventions used throughout the paper.

\paragraph{Pieces.} Intervals are half-open, $[a,b)$, except the interval ending at $1$ (which includes that endpoint). Therefore, cutting an interval at a point $x$ yields the disjoint intervals $[a,x)$ and $[x,b)$. Every piece comes with a list of its intervals in a specified order, which need not be their order on $[0,1]$; e.g., in \Core{} the order is left to right, while \Repair{} specifies the order of its pieces explicitly. Within each protocol call, we relabel the participating agents as $1,\dots,n$.

\paragraph{Prefixes, suffixes and prefix cuts.} 
A \emph{prefix} of a piece consists of its first $k-1$ intervals and an initial part $[a,x)$ of its $k$-th interval $[a,b)$, for some $k$ and some $x\in[a,b]$; a \emph{suffix} is what remains after removing a prefix (the whole piece and the empty piece are suffixes of every piece). The value of a prefix of a piece $P$ changes continuously with $x$, so for every $0\le v\le V_i(P)$ some prefix is worth exactly $v$ to agent $i$; a \emph{prefix cut} on $P$ for agent $i$ and value $v$ finds the smallest such prefix. 

\paragraph{Recorded values.}
Every piece our protocols handle is a union of intervals in the partition determined by $0$, $1$, and all previously returned \textsc{Cut} points. We record each agent's value of every such interval; subroutine and recursive calls inherit these records for all (participating) agents. Initially, every agent's value of $[0,1]$ is $1$. Whenever a new cut point splits an interval, one \textsc{Eval} query per participating agent gives the value of one part, and subtraction gives the value of the other. Therefore a prefix cut costs at most $n+1$ queries: one \textsc{Cut} query in the interval identified using recorded values, and at most $n$ \textsc{Eval} queries. Bundle and residue values are sums of recorded interval values, so all remaining operations require no additional queries.

\section{Reduction to Domination}\label{sec: reduce to domination}

In this section, we state our main protocol, \textsc{Divide}. We reduce the problem of finding an envy-free allocation to the problem of separating the agents into two sets, such that every agent in the first set \emph{dominates} every agent in the second set. Intuitively, if $i$ dominates $j$ in a partial allocation, then $i$ will not envy $j$ even if $j$ receives the entire residue.
Domination is formally defined as follows.

\begin{definition}[Domination]
In a partial allocation $X$ with residue $R$, we say that agent $i$ \emph{dominates} agent $j$ if $V_i(X_i) \geq V_i(X_j) + V_i(R)$.
\end{definition}

If partial allocations $X$ and $Y$ allocate disjoint parts of the cake, $X\oplus Y$ denotes the allocation that gives $X_i\cup Y_i$ to every agent $i$; if $X$ and $Y$ are envy-free among the same set of agents, so is $X\oplus Y$. For a partial allocation $X$ with residue $R$, a \emph{dominated group} of $(X,R)$ is a set of agents $E$ with $\emptyset\neq E\subsetneq N$ such that every agent in $N\setminus E$ dominates every agent in $E$.

Our main protocol (\Cref{algo:divide}) first calls $\Dominate$ to get a partial envy-free allocation where a group of agents $E$ (that is nonempty and a proper subset of $N$) is dominated by everyone not in this set, and then recursively divides the residue among agents in $E$.  This domination-based reduction is also used in the $n$-agent protocols of Aziz and Mackenzie~\cite{aziz2016discrete} and Ye and Bai~\cite{ye2026cutting}.

\begin{algorithm}[H]
\caption{$\Divide(N,C)$}\label{algo:divide}
\If{$|N|=1$}{give $C$ to the agent and return\;}
\If{$|N|=2$}{Run \textsc{Cut-And-Choose}: the first agent cuts $C$ into two pieces of equal value to her, the second chooses her favorite piece, and the first gets the other; return\;}
$(X,R,E) \gets \Dominate(N,C)$\;\label{divide:line:dominate call}
$Y \gets \Divide(E,R)$\; \label{divide:line:recurse}
Return $X\oplus Y$, with $Y_i=\emptyset$ for $i\notin E$\;
\end{algorithm}

Note that in recursive calls (i.e., Line~\ref{divide:line:recurse}), the cake may be a finite union of intervals, and the valuations need not be normalized, or might be zero; the guarantee for $\Dominate$ below applies to such inputs. We use the following lemma, which we prove in~\Cref{sec: reduce to core and repair}.

\begin{restatable}{lemma}{dominatemain}\label{lem: dominate}
    For a set $N$ of $n \geq 3$ agents and a cake $C$, $\Dominate(N,C)$ returns a partial envy-free allocation $X$ of $C$ and a set of agents $\emptyset \ne E \subsetneq N$, such that every agent $i \notin E$ dominates every agent $j \in E$. Furthermore, $\Dominate{}$ uses at most $32^n n^7(n!)^2$ Robertson--Webb queries.
\end{restatable}

Given~\Cref{lem: dominate}, the proof of our main result is straightforward.

\begin{theorem}\label{thm: main theorem}
   For a set $N$ of $n$ agents and a cake $C$, $\Divide(N,C)$ returns a complete envy-free allocation of $C$, and uses at most $2\cdot32^n n^7(n!)^2 \in 2^{O(n \log n)}$ Robertson--Webb queries.
\end{theorem}

\begin{proof}
We prove correctness by induction on the number of agents.
For one or two agents, the statement is trivial. Otherwise, by~\Cref{lem: dominate}, Line~\ref{divide:line:dominate call} returns a partial envy-free allocation $X$, a residue $R$, and a nonempty, proper subset of agents $E \subsetneq N$, such that every agent $i \notin E$ dominates every agent $j \in E$. By induction, Line~\ref{divide:line:recurse} returns a complete and envy-free allocation $Y$ of the residue $R$, among agents in $E$.  The final allocation of \textsc{Divide} is $X_i$ for all $i \notin E$ and $X_i \cup Y_i$ for all $i \in E$. This allocation is complete (since the recursive call leaves no residue) as well as envy-free. First, for all $i \in E$, $X_i$ is weakly preferred to $X_j$ for all $j \in N$, $j \neq i$, so there is no envy towards agents not in $E$. By induction, $Y_i$ is weakly preferred to $Y_j$, for all $j \in E$, $j \neq i$, and therefore $X_i \cup Y_i$ is weakly preferred to $X_j \cup Y_j$, so there is no envy towards agents in $E$. Therefore, agents in $E$ do not envy anyone. Second, for all $i \notin E$, $X_i$ is weakly preferred to $X_j$ for all $j \in N$, $j \neq i$, so there is no envy towards agents not in $E$. Finally, $i$ dominates $j$, for all $j \in E$, and therefore $V_i(X_i) \geq V_i(X_j) + V_i(R) \geq V_i(X_j) + V_i(Y_j)$, so there is no envy towards agents in $E$.

$\Dominate$ is called at most once for each number $m$ of remaining agents, where $3 \leq m \leq n$. By~\Cref{lem: dominate}, such a call uses at most $32^m m^7(m!)^2$ queries. The final call uses no queries if one agent remains, and two queries if two agents remain. Therefore, for $n\ge2$, the total number of queries is at most
$2+\sum_{m=3}^n 32^m m^7(m!)^2
\le n^7(n!)^2\sum_{m=2}^n32^m
<2\cdot32^n n^7(n!)^2$.
\end{proof}

\section{Reduction to \Core~and \Repair}\label{sec: reduce to core and repair}

In this section, we prove~\Cref{lem: dominate} by implementing $\Dominate(N,C)$, where $n=|N|\geq 3$. 

The protocol maintains an envy-free partial allocation $X$ and its residue $R$; initially $X_i=\emptyset$, for every agent $i$, and $R=C$. In each iteration, we are either done (i.e., there is a nonempty proper group dominated by every outsider), or we modify the allocation so that some agent's \emph{domination count} (the number of agents she dominates) increases, without decreasing any other agent's count. Since the sum of these counts is at most $n^2$, at most $n^2$ iterations do not return.

\paragraph{Subroutines.} The protocol uses two more protocols as subroutines: \Core~and \Repair. 

\Core~receives as input a cake $C$, a set of agents $N$, and a designated cutter $c \in N$; it outputs a partial envy-free allocation. We use the following lemma, which we prove in~\Cref{sec:core and repair}.

\begin{restatable}{lemma}{coremain}\label{lem:core main}
For a set $N$ of $n\geq 2$ agents, a cake $C$ (possibly empty, and possibly of value zero to some agents), and a designated cutter $c\in N$, $\Core(c,N,C)$ returns a partial envy-free allocation $Y$ of $C$, with residue $U$, in which the cutter receives exactly $V_c(Y_c)=V_c(C)/n$. Furthermore, $\Core{}$ uses at most $n^3(n!)^2$ Robertson--Webb queries.
\end{restatable}

\Repair~receives as input a partial ``almost envy-free'' allocation $D$ with residue $R$, a set of agents $N$ and a parameter $\eps > 0$. It outputs an envy-free allocation $W$, by adding small pieces from $R$ to the bundles of $D$ and reassigning the resulting bundles (so, e.g., $W_5$ could be $D_3$ plus a small piece of $R$). We use the following lemma, which we prove in~\Cref{sec:core and repair}.

\begin{restatable}{lemma}{repairmain}\label{lem:repair main}
Let $N$ be a set of $n\geq 2$ agents and $0<\eps\leq 1/(6n^3)$. Let $D=(D_i)_{i\in N}$ be a partial allocation with residue $R$,
with $V_i(R)>0$ for all $i\in N$, such that
$V_i(D_i)\geq V_i(D_j)-\eps V_i(R)$ for all $i,j\in N$. 
Then $\Repair(N,D,R,\eps)$ returns an envy-free partial allocation $W$ by assigning the bundles $D_j \cup Z_j$, one per agent, where the pieces $Z_j\subseteq R$ are pairwise disjoint and satisfy $V_i(Z_j)\leq 3n^2\eps V_i(R)$ for all $i,j\in N$. The residue is $R\setminus \left( \bigcup_{i \in N}Z_i \right)$. Furthermore, $\Repair{}$ uses at most $\frac{n}{\eps}+n^3(n!)^2$ Robertson--Webb queries.
\end{restatable}

Given~\Core{} and~\Repair{}, we implement $\Dominate(N,C)$ as~\Cref{alg:dominate}.

\begin{algorithm}[H]
\caption{$\Dominate(N,C)$}\label{alg:dominate}
Fix any agent $c\in N$; set $X_i=\emptyset$ for all $i \in N$, and $R=C$\;
\While{true}{
    $(Y,U)\gets\Core(c,N,R)$\;\label{dominate:line:core}
    \If{$(X\oplus Y,U)$ has a dominated group $E$\label{dominate:line:test}}{
        \Return{$(X\oplus Y,U,E)$}\;\label{dominate:line:return}
    }
    Choose $\alpha\in\{(4n^2)^{-t}:2\le t\le n+1\}$, agents $i^* \ne j^*$, and a permutation $\pi$ of $N$ s.t.\newline
    \hspace*{1.5em}(S1) $i^*$ does not dominate $j^*$ in $(X,R)$,\newline
    \hspace*{1.5em}(S2) $V_i(Y_{\pi(i)})\ge V_i(Y_i)-\alpha V_i(U)$ for all $i\in N$, and\newline
    \hspace*{1.5em}(S3) $V_{i^*}(Y_{\pi(j^*)})\le V_{i^*}(Y_{i^*})-4n^2\alpha V_{i^*}(U)$\;\label{dominate:line:choose}
    $(W,R')\gets\Repair\bigl(N,(X_i\cup Y_{\pi(i)})_{i\in N},U,\alpha\bigr)$\;\label{dominate:line:repair}
    \While{$V_{i^*}(R')>\alpha V_{i^*}(U)$}{\label{dominate:line:shrink}
        $(F,R')\gets\Core(i^*,N,R')$, and $W\gets W\oplus F$\;\label{dominate:line:shrink:core}
    }
    $(X,R)\gets(W,R')$\;\label{dominate:line:rename}
}
\end{algorithm}

We start by proving that Line~\ref{dominate:line:choose} can be executed whenever the protocol reaches it, i.e., there exist $\alpha$, $i^*$, $j^*$, and $\pi$ that satisfy properties (S1)--(S3).

\begin{lemma}\label{lem:alpha pi exist for dominate}
Let $X$ be an envy-free partial allocation with residue $R$, let $(Y,U)=\Core(c,N,R)$, and suppose that $(X\oplus Y,U)$ has no dominated group. Then $V_i(U)>0$ for every $i \in N$, and there exist $\alpha\in\{(4n^2)^{-t}:2\le t\le n+1\}$, agents $i^* \ne j^*$, and a permutation $\pi$ of $N$ satisfying (S1)--(S3).
\end{lemma}

\begin{proof}
If $V_i(U)=0$ for some agent $i$, then $i$ dominates every other agent in $(X\oplus Y,U)$, since $X\oplus Y$ is envy-free (i.e., $V_i( (X\oplus Y)_i ) \geq V_i( (X\oplus Y)_j ) = V_i( (X\oplus Y)_j ) + V_i(U)$, for all agents $j \neq i$), so $N\setminus\{i\}$ would be a dominated group. Therefore, $V_i(U)>0$ for every $i$. 

Let
\[
g_{ij}=\frac{V_i(Y_i)-V_i(Y_j)}{V_i(U)}.
\]
Then, $g_{ij} \geq 0$ for $i \neq j$, since $Y$ is envy-free, and $g_{ii}=0$. For (S2), it suffices to show that every agent $i$ has $g_{i\pi(i)}\le\alpha$; for (S3), it suffices to show that $g_{i^*\pi(j^*)}\ge 4n^2\alpha$. 

Let $G$ be the directed graph on $N$ with an edge $i\to j$ whenever $i\ne j$ and $i$ does \emph{not} dominate $j$ in $(X,R)$; the edges of $G$ are exactly the pairs satisfying (S1). Observe that $G$ is strongly connected: otherwise $G$ would have a non-empty sink strongly connected component $S$, and every agent in $S$ would dominate every agent outside $S$ in $(X,R)$ (and therefore also in $(X\oplus Y,U)$, since adding the envy-free allocation $Y$ does not decrease any difference $V_i(X_i)-V_i(X_j)$, and $U\subseteq R$); then $N\setminus S$ would be a dominated group. We consider two cases:
\begin{enumerate}[leftmargin=*]
    \item Some edge $i \to j$ of $G$ has $g_{ij}\ge (4n^2)^{-n}$. Take $\alpha=(4n^2)^{-(n+1)}$, $i^* = i$, $j^* = j$, and $\pi$ the identity permutation (i.e., $\pi(k) = k $). 
    
    (S1) holds because $i^* \to j^*$ is an edge of $G$, and (S2) holds trivially since $g_{kk}=0\le\alpha$. Finally, $g_{i^*j^*} \ge (4n^2)^{-n} = 4n^2\alpha$ so (S3) holds as well.
    \item Every edge $i\to j$ of $G$ has $g_{ij}<(4n^2)^{-n}$. By~\Cref{lem:core main}, $nV_c(Y_c)=V_c(R)$, while $\sum_{j\in N}V_c(Y_j)=V_c(R)-V_c(U)$ because $Y$ and $U$ partition $R$. Therefore, $\sum_{j\ne c}g_{cj}= \frac{nV_c(Y_c)-\sum_{j\in N}V_c(Y_j)}{V_c(U)}=1$. Since this is a sum of $n-1$ nonnegative terms, some $\hat{j} \ne c$ satisfies $g_{c\hat{j}}\ge 1/(n-1)>(4n^2)^{-1}$. The $n-1$ disjoint intervals $( (4n^2)^{-(t+1)},(4n^2)^{-t}]$, for $t=1,\dots,n-1$, cover $((4n^2)^{-n},(4n^2)^{-1}]$. Of the $n$ numbers $g_{1\hat{j}}, g_{2\hat{j}},\dots,g_{n\hat{j}}$, $g_{c\hat{j}} > (4n^2)^{-1}$ is above all these intervals and $g_{\hat{j}\hat{j}} = 0$ is below them; therefore, the remaining $n-2$ values $g_{i\hat{j}}$ must leave some interval empty. That is, there is some $t^*$ such that no $g_{i\hat{j}}$ lies in $((4n^2)^{-(t^*+1)},(4n^2)^{-t^*}]$. Let $\alpha=(4n^2)^{-(t^*+1)}$; then $(4n^2)^{-n}\le\alpha\le(4n^2)^{-2}$, and every agent $i$ satisfies either $g_{i\hat{j}}\le\alpha$ or $g_{i\hat{j}}>4n^2\alpha$. Let $L=\{i:\ g_{i\hat{j}}\le\alpha\}$ and $H=N\setminus L=\{i:\ g_{i\hat{j}}>4n^2\alpha\}$; by definition, $\hat{j} \in L$ and $c\in H$. Since $G$ is strongly connected and $H,L$ are nonempty, there is an edge $i^* \to j^*$ with $i^* \in H$ and $j^* \in L$. Note that $j^* \ne \hat{j}$: otherwise, by the definition of $H$, the edge $i^* \to \hat{j}$ would satisfy $g_{i^* \hat{j}} > 4n^2\alpha > (4n^2)^{-n}$ (as $\alpha\ge(4n^2)^{-n}$), contradicting the assumption of this case.
    Finally, let $\hat{j} = u_0 \rightarrow u_1 \rightarrow \dots \rightarrow u_{\ell} = j^*$ be a simple path in $G$ (i.e., no repeated vertices; such a path exists since $G$ is strongly connected) from $\hat{j}$ to $j^*$, and let $\pi$ be the permutation of $N$ that cycles the vertices along this path and fixes every other agent: $\pi(j^*)=\hat{j}$, $\pi(u_s)=u_{s+1}$ for $s<\ell$, and $\pi(i)=i$ otherwise.

    We have that (S1) holds because $i^* \to j^*$ is an edge of $G$. (S2) holds because: (i) $g_{j^*\pi(j^*)} = g_{j^*\hat{j}} \leq \alpha$ ($j^* \in L$), (ii) $g_{u_s \pi(u_s)} = g_{u_s u_{s+1}} < (4n^2)^{-n} \leq \alpha$ ($u_s \to u_{s+1}$ is an edge of $G$), and (iii) $g_{i \pi(i)} = g_{ii} = 0$ for all other agents $i$. Finally, (S3) holds since $g_{i^* \pi(j^*)} = g_{i^*\hat{j}} > 4n^2\alpha$ ($i^* \in H$).

\end{enumerate}

\end{proof}

We are now ready to prove~\Cref{lem: dominate}. We first restate the lemma.

\dominatemain*

\begin{proof}[Proof of \Cref{lem: dominate}]
For an envy-free partial allocation $X$ with residue $R$ and an agent $i$, let
\[
\Dom_i(X,R)=\{j \in N:\ V_i(X_i)-V_i(X_j)\ge V_i(R)\}.
\]
For $j\ne i$, $j\in\Dom_i(X,R)$ says that $i$ dominates $j$.

We analyze one iteration of~\Dominate{}, assuming that $X$ is envy-free at its start; this holds initially because every $X_i$ is empty. We show that the iteration either returns a correct output or ends with an envy-free allocation $X$ and a larger $\sum_{i\in N}|\Dom_i(X,R) |$; since $\sum_{i\in N}|\Dom_i(X,R) | \leq n^2$,~\Dominate{} will therefore terminate with a correct output after a certain number of iterations.

Since \Core{} returns an envy-free allocation $Y$, the allocation $X\oplus Y$ is envy-free. Therefore, if the protocol returns in Line~\ref{dominate:line:return}, its output has all required properties. Otherwise, let $\alpha,i^*,j^*,\pi$ satisfy (S1)--(S3); such choices exist by~\Cref{lem:alpha pi exist for dominate}.

\paragraph{The \Repair{} call (Line~\ref{dominate:line:repair}) is valid.}
Let $D=(D_i)_{i \in N}$, where $D_i=X_i\cup Y_{\pi(i)}$, be the allocation passed to \Repair{} in Line~\ref{dominate:line:repair}; it is a partial allocation with residue $U$, and $V_i(U)>0$ for every $i$ (\Cref{lem:alpha pi exist for dominate}). Also $\alpha\le(4n^2)^{-2}<1/(6n^3)$. Finally, $D$ is almost envy-free: for all $i,j\in N$, we have
\begin{align*}
V_i(D_j) &=V_i(X_j)+V_i(Y_{\pi(j)}) \\
&\le V_i(X_i)+V_i(Y_i) \tag{$X$ and $Y$ are envy-free} \\
&\le V_i(X_i)+V_i(Y_{\pi(i)})+\alpha V_i(U) \tag{(S2) for agent $i$} \\
&=V_i(D_i)+\alpha V_i(U). 
\end{align*}
Therefore, \Cref{lem:repair main} applies. Let $(W,R')$ be the output of \Repair{} in Line~\ref{dominate:line:repair}, and let $\sigma(j)$ denote the agent who receives the bundle $D_j\cup Z_j$ (i.e., $W_{\sigma(j)}=D_j\cup Z_j$), where $\sigma$ is a permutation of $N$ (\Repair{} assigns one bundle per agent).

\paragraph{Agent $i^*$ gains an advantage over the bundle indexed by agent $j^*$.}
We have
\begin{align*}
  V_{i^*}(W_{i^*}) &\ge V_{i^*}(W_{\sigma(i^*)}) \tag{$W$ is envy-free} \\ 
  &\ge V_{i^*}(D_{i^*}) \tag{$D_{i^*}\subseteq W_{\sigma(i^*)}$} \\ 
  &= V_{i^*}(X_{i^*})+V_{i^*}(Y_{\pi(i^*)}) \\ 
  &\ge V_{i^*}(X_{i^*})+V_{i^*}(Y_{i^*})-\alpha V_{i^*}(U). \tag{(S2) for agent $i^*$}  
\end{align*}
On the other hand, $W_{\sigma(j^*)}=X_{j^*}\cup Y_{\pi(j^*)}\cup Z_{j^*}$, where: (i) $V_{i^*}(X_{j^*})\le V_{i^*}(X_{i^*})$, since $X$ is envy-free, (ii) $V_{i^*}(Y_{\pi(j^*)})\le V_{i^*}(Y_{i^*})-4n^2\alpha V_{i^*}(U)$ by (S3), and (iii) $V_{i^*}(Z_{j^*})\le 3n^2\alpha V_{i^*}(U)$, by~\Cref{lem:repair main}; therefore, $V_{i^*}(W_{\sigma(j^*)}) \le V_{i^*}(X_{i^*})+V_{i^*}(Y_{i^*}) - n^2\alpha V_{i^*}(U)$. Subtracting from the previous lower bound on $V_{i^*}(W_{i^*})$ we have
\[
V_{i^*}(W_{i^*})-V_{i^*}(W_{\sigma(j^*)})\ge (n^2-1)\alpha V_{i^*}(U)>0.
\]

\paragraph{Repeated \Core{} calls create domination.}
Next, we argue that the loop starting in Line~\ref{dominate:line:shrink} terminates after a bounded number of calls and leaves $i^*$ dominating agent $\sigma(j^*)$.

Each pass of the loop calls \Core{} on $R'$ (which returns an envy-free allocation $F$ of $R'$ and makes $R'$ its residue) and replaces $W$ by $W\oplus F$; this keeps $W$ envy-free and does not decrease any difference $V_i(W_i)-V_i(W_j)$. In particular, the advantage $V_{i^*}(W_{i^*})-V_{i^*}(W_{\sigma(j^*)})\ge (n^2-1)\alpha V_{i^*}(U)$ persists.
In each call, cutter $i^*$ receives a piece worth $1/n$ of her value of the current residue, by~\Cref{lem:core main}. Therefore, her value of the remaining residue is at most $(n-1)/n$ of its previous value, and after $t$ calls we have $V_{i^*}(R')\le(1-\frac1n)^tV_{i^*}(U)\le e^{-t/n}V_{i^*}(U)$. The loop terminates after at most $\lceil n\ln(1/\alpha)\rceil$ calls, which is at most $\lceil n(n+1)\ln(4n^2)\rceil$ as $\alpha\ge(4n^2)^{-(n+1)}$. At termination $V_{i^*}(R')\le\alpha V_{i^*}(U)$, so
\begin{equation}\label{eq: i dominates sigma j}
 V_{i^*}(W_{i^*})-V_{i^*}(W_{\sigma(j^*)})\ \ge\ (n^2-1)\alpha V_{i^*}(U)\ \ge\ \alpha V_{i^*}(U)\ \ge\ V_{i^*}(R'),   
\end{equation}
and therefore agent $i^*$ dominates agent $\sigma(j^*)$. 

\paragraph{$\sum_i |\Dom_i|$ increases.}
We compare $(X,R)$, the allocation and residue at the start of the iteration, with $(W,R')$ at the end of the loop, which Line~\ref{dominate:line:rename} makes the new $(X,R)$. We first show that old dominations are inherited: for all agents $i$ and $j$, if $j\in\Dom_i(X,R)$ then $\sigma(j)\in\Dom_i(W,R')$. Second, we show that agent $i^*$'s domination count increases: $\sigma(j^*)\in\Dom_{i^*}(W,R')$, but $j^* \notin \Dom_{i^*}(X,R)$. Since $\sigma$ is a permutation, distinct agents $j$ have distinct images $\sigma(j)$; so the first fact gives $|\Dom_i(W,R')|\ge|\Dom_i(X,R)|$ for every agent $i$, and, since $\sigma(j^*)$ is not the image of any $j\in\Dom_{i^*}(X,R)$, the second fact gives $|\Dom_{i^*}(W,R')|\ge|\Dom_{i^*}(X,R)|+1$. Therefore, after Line~\ref{dominate:line:rename}, the iteration ends with an envy-free allocation $X$ and a larger $\sum_{i\in N}|\Dom_i(X,R)|$.

For the second fact, by Inequality~\eqref{eq: i dominates sigma j} we have that $\sigma(j^*)\in\Dom_{i^*}(W,R')$, while $j^*\notin\Dom_{i^*}(X,R)$ by (S1).
For the first fact, every old bundle $X_j$ is contained in $D_j\cup Z_j = W_{\sigma(j)}$ after \Repair{}, and the loop only makes bundles in $W$ larger, so $X_j\subseteq W_{\sigma(j)}$ at the end of the loop. Everything added to the bundles during the iteration comes from $R$ (and lies outside $R'$), so $W_{\sigma(j)}\subseteq X_j\cup(R\setminus R')$ and $V_i(W_{\sigma(j)})\le V_i(X_j)+V_i(R\setminus R')$ for all agents $i,j$. Moreover, $V_i(W_i)\ge V_i(W_{\sigma(i)})\ge V_i(X_i)$, since $W$ is envy-free and $X_i\subseteq W_{\sigma(i)}$. Therefore, for all agents $i$ and $j$,
\[
V_i(W_i)-V_i(W_{\sigma(j)})\ \ge\ V_i(X_i)-V_i(X_j)-V_i(R\setminus R').
\]
If $j\in\Dom_i(X,R)$ the RHS is at least $V_i(R)-V_i(R\setminus R')=V_i(R')$, i.e., $\sigma(j)\in\Dom_i(W,R')$.

\paragraph{Termination and query count.}
Each $\Dom_i(X,R)$ is a subset of $N$, so $\sum_{i\in N}|\Dom_i(X,R)|\le n^2$ at all times; as the sum increases in every iteration that does not return, at most $n^2$ iterations do not return. Each iteration makes one \Core{} call in Line~\ref{dominate:line:core}; each iteration that does not return also makes one \Repair{} call, with $\eps=\alpha\ge(4n^2)^{-(n+1)}$, and at most $\lceil n(n+1)\ln(4n^2)\rceil$ \Core{} calls in Line~\ref{dominate:line:shrink:core}. Thus the protocol makes at most $n^2$ \Repair{} calls and at most $n^2\bigl(1+\lceil n(n+1)\ln(4n^2)\rceil\bigr)+1 \leq 6 n^4 \ln(n)$ \Core{} calls. The searches in Lines~\ref{dominate:line:test}, \ref{dominate:line:choose} and~\ref{dominate:line:shrink} use only values of bundles and residues, which are sums of stored interval values, so they require no additional queries (see~\Cref{sec:conventions}). Combining these call counts with~\Cref{lem:core main,lem:repair main}, the total number of queries is at most
\begin{align*}
&n^2\left(n(4n^2)^{n+1}+n^3(n!)^2\right)
 +6n^7\ln n\,(n!)^2\\
&\qquad=n^3(4n^2)^{n+1}
 +(n^5+6n^7\ln n)(n!)^2
 \le 32^n n^7(n!)^2,
\end{align*}
where the last inequality uses $n!\ge(n/e)^n$ and $n\ge3$.
\end{proof}

\section{Reduction to \Permute}\label{sec:core and repair}

In this section, we prove \Cref{lem:core main,lem:repair main} by implementing \Core{} and \Repair. Both use one last subroutine, \Permute, which receives as input a set of agents $N$ and $n$ pairwise disjoint pieces $P_1,\dots,P_n$ (some possibly empty), and outputs a partial envy-free allocation $Y$ of $P_1\cup\dots\cup P_n$, with residue $R$. We use the following lemma, which we prove in~\Cref{sec:permute}.

\begin{restatable}{lemma}{permutelemma}\label{lem:permute}
For a set $N$ of $n$ agents and $n$ pairwise disjoint pieces $P_1,\dots,P_n$ (some possibly empty), $\Permute(N,(P_i)_{i \in N})$ returns a partial envy-free allocation $Y$ of $P_1\cup\dots\cup P_n$, with residue $R$, in which, for some bijection $\pi$ from $N$ to $\{1,\dots,n\}$, every $Y_i$ is a suffix of $P_{\pi(i)}$. Moreover, (i) some agent $i^*$ receives her piece $P_{\pi(i^*)}$ whole, and (ii) for every nonempty set $S\subseteq N$ of agents with $Y_i \ne P_{\pi(i)}$ for all $i \in S$, there are agents $i^* \in S$ and $j^* \notin S$ with $V_{j^*}(Y_{i^*})=V_{j^*}(Y_{j^*})$. Furthermore, $\Permute{}$ makes at most $\binom n2(n!)^2$ prefix cuts.
\end{restatable}

\subsection{Implementing \Core}

In this section, we state \Core{} and prove~\Cref{lem:core main}. Stating \Permute{} as a separate routine allows for a remarkably simple description of \Core{}: The cutter partitions $C$ into $n$ equally valued pieces (using $n-1$ prefix cuts), and \Permute{} allocates to each agent a suffix of a different piece.

\begin{algorithm}[H]
\caption{$\Core(c,N,C)$, where $n=|N|$}\label{alg:core}
$c$ partitions $C$ into pieces $P_1,\dots,P_n$ worth $V_c(C)/n$ each (using $n-1$ prefix cuts) \;\label{core:cut}
\Return{$\Permute(N,(P_i)_{i \in N})$}\;\label{core:permute}
\end{algorithm}

We restate the lemma; its proof is short.

\coremain*

\begin{proof}[Proof of \Cref{lem:core main}]
Let $(Y,R)$ be the output of \Permute{} in Line~\ref{core:permute}; by \Cref{lem:permute}, it is a partial envy-free allocation of $P_1\cup\dots\cup P_n=C$ with residue $R$, and some agent receives a piece $P_i$ whole. Since $Y$ is envy-free, $V_c(Y_c)\ge V_c(P_i)=V_c(C)/n$. But, since $Y_c$ is a suffix of a single piece $P_{j}$, $V_c(Y_c)\le V_c(P_{j}) = V_c(C)/n$; therefore, overall, $V_c(Y_c) = V_c(C)/n$, as needed. Finally, \Core{} makes $n-1$ prefix cuts in Line~\ref{core:cut} and, by \Cref{lem:permute}, at most $\binom n2(n!)^2$ prefix cuts in Line~\ref{core:permute}. Since we need at most $n+1$ queries per prefix cut, this is at most $(n+1)\bigl(n-1+\binom n2(n!)^2\bigr)\le n^3(n!)^2$ queries.
\end{proof}

\subsection{Implementing \Repair}

In this section, we state \Repair{} and prove~\Cref{lem:repair main}. Recall that \Repair{} takes as input a set of agents $N$, an ``almost envy-free'' partial allocation $D=(D_i)_{i\in N}$ with residue $R$, and a parameter $\eps>0$; it outputs an envy-free partial allocation by adding small pieces of $R$ to the bundles of $D$ (and possibly reassigning the resulting bundles). Every old bundle $D_i$ remains intact. 

For disjoint pieces $A$ and $B$, let $A\cat B$ be their concatenation: the interval list of $A$ followed by that of $B$. \Repair{} constructs pairwise disjoint pieces $T_i\subseteq R$ and calls \Permute{} on the pieces $P_i=T_i\cat D_i$. We give \Permute{} enough cake to trim without ever cutting into an old bundle. The added cake is first because \Permute{} removes prefixes and returns suffixes: it must trim through $T_i$ before reaching $D_i$. We choose $T_i$ so that the trimming never reaches an old bundle $D_i$.

\Repair{} works as follows. Let $M=\lceil1/(3n\eps)\rceil$. Every agent marks the $M-1$ points that split $R$ into $M$ pieces of equal value to her (using $M-1$ prefix cuts). Cut $R$ at all $n(M-1)$ marks; we call the resulting pieces \emph{cells} (if $R$ is not a single interval, a cell may consist of several intervals of $R$). Every cell lies inside one piece of every agent's partition, so it is worth at most $V_i(R)/M$ to every agent $i$.\footnote{This is a standard discretization trick in cake cutting; see, e.g.,~\cite{balkanski2014simultaneous,branzei2022query}. For example, asking every agent to mark the points that split the cake into $\lceil1/\eps\rceil$ pieces of equal value to her, treating the resulting cells as indivisible items and applying any EF1 algorithm gives an $\eps$-EF cake cutting algorithm with $O(n^2/\eps)$ Robertson--Webb queries.} Next, run round robin for $n-1$ rounds: in each round, every agent, in a fixed order, chooses her favorite available cell. Finally, put each agent's chosen cells before her old bundle, in the order she chose them, and call \Permute{} on the resulting pieces.

\begin{algorithm}[H]
\caption{$\Repair(N,D,R,\eps)$, where $n=|N|$}\label{alg:repair}
Set $M=\lceil1/(3n\eps)\rceil$. Every agent marks the $M-1$ points
that split $R$ into $M$ pieces of equal value to her.
Cut $R$ at all the marks to form the cells described above\;\label{line:rep:cells}
\RRmark{repair-rr-top}
\RRFor{$r=1,\ldots,n-1$}{
  \For{each agent $i = 1, \dots, n$}{
    Let $Q_{i,r}$ be her favorite available cell; mark it as  unavailable\;\label{rep:choose}
  }
}
\tikz[remember picture,overlay]{%
  \draw[decorate,decoration={brace,amplitude=4pt}]
    ($(repair-rr-top)+(0.75\linewidth,1.3ex)$) --
    ($(repair-rr-top |- repair-rr-bottom)+(0.75\linewidth,-0.5ex)$)
    node[midway,right=8pt,align=center,font=\small] {Round robin\\for $n-1$ rounds};
}%
For every $i \in N$, form
$P_i=Q_{i,1}\cat\cdots\cat Q_{i,n-1}\cat D_i$\;\label{rep:pieces}
$(W,U)\gets\Permute(N,(P_i)_{i \in N})$\;\label{rep:permute}
\Return{allocation $W$ and residue $(R\cup (\bigcup_{i \in N} D_i) )\setminus\bigcup_{i \in N} W_i$}\;
\end{algorithm}

We are now ready to prove~\Cref{lem:repair main}. For ease of exposition, we restate the lemma.

\repairmain*

\begin{proof}[Proof of~\Cref{lem:repair main}]

The query count follows directly. Forming the cells uses $n(M-1)<1/(3\eps)$ prefix cuts, and \Permute{} uses at most $\binom n2(n!)^2$ prefix cuts (\Cref{lem:permute}). All other operations use recorded values. Under our ``recorded value'' convention (see~\Cref{sec:conventions}), each prefix cut costs at most $n+1$ queries, giving a total of at most $(n+1)\left(\frac{1}{3\eps}+\binom n2(n!)^2\right) \le \frac{n}{\eps}+n^3(n!)^2$ queries.

Line~\ref{line:rep:cells} makes at least $M-1$ marks, and therefore there are at least $M \geq 1/(3n \eps) \geq 2n^2$ cells, which is at least the number of cells needed for the round robin step (therefore, every line of the protocol is successfully executed). By~\Cref{lem:permute}, the call in Line~\ref{rep:permute} returns an envy-free allocation in which every agent receives a suffix of a different piece $P_i$. It remains to show that the added pieces satisfy the stated value bound ($V_i(Z_j)\leq 3n^2\eps V_i(R)$), and that  every old bundle remains intact.

\paragraph{Properties of the chosen cells.}
Line~\ref{line:rep:cells} makes at most $n(M-1)$ marks, so there are at most $n(M-1)+1 < 1/(3\eps)+1 \leq 1/(2\eps)$ cells ($\eps\le 1/(6n^3)\le1/6$). At any time during the $n-1$ rounds of round robin, at most $n(n-1)\le n^2$ cells have already been picked. Each cell is worth at most $V_i(R)/M\le 3n\eps V_i(R)$ to every agent $i$, so, for every agent $i$, the cells already picked are (in total) worth at most $3n^3\eps V_i(R)\le \frac{1}{2} V_i(R)$, and therefore the cells still available are (in total) worth at least $\frac{1}{2}V_i(R)$. For each agent $i$, let $Z_i$ be the portion of her chosen cells retained by \Permute{}. These pieces are pairwise disjoint subsets of $R$. Since each agent picks $n-1$ cells, for every pair of agents $i,j$, we have $V_i(Z_j)\le\sum_{r=1}^{n-1}V_i(Q_{j,r}) \le3n(n-1)\eps V_i(R) \le3n^2\eps V_i(R)$.
This proves the required bound on the added pieces. It remains to show that every old bundle remains intact (in the allocation returned by \Permute{}).

\paragraph{Levels.}
Let $W$ be the allocation returned by \Permute{} in Line~\ref{rep:permute}. By~\Cref{lem:permute}, for some bijection $\pi$, every $W_i$ is a suffix of $P_{\pi(i)}$. For every $i \in N$, let $S_i$ be the suffix of $P_i$ that is allocated, i.e., $S_i = W_j$ for the agent $j$ with $\pi(j)=i$; we say that agent $j$ \emph{holds} $S_i$. Every agent holds exactly one of the pieces $S_i$, and every $S_i$ is held by exactly one agent. We define the \emph{level} $\ell(i)$ of $S_i$ as the number of cells of $Q_{i,1}\cat\cdots\cat Q_{i,n-1}$ contained in $S_i$, or $\ell(i) = -1$ if $D_i \not\subseteq S_i$. Since $S_i$ is a suffix of $Q_{i,1}\cat\cdots\cat Q_{i,n-1}\cat D_i$, if $\ell(i)\ge0$, then $S_i$ contains $D_i$ and the last $\ell(i)$ cells of $Q_{i,1}\cat\cdots\cat Q_{i,n-1}$. We will show that $\ell(i) \ge 0$ for all $i \in N$, i.e., that no old bundle is cut.

\paragraph{Higher levels are preferred.}
We claim that, for all $i,j\in N$ with $\ell(i)\ge\ell(j)+2$,
\begin{equation}\label{eq:levels}
V_i(S_i)>V_i(S_j).
\end{equation}
Fix such $i,j$. When agent $i$ picked $Q_{i,r-1}$, the cell $Q_{j,r}$ was still available, so $V_i(Q_{i,r-1})\ge V_i(Q_{j,r})$ for every $2\le r \le n-1$. Since $S_j$ is contained in the union of $D_j$ and the last $\ell(j)+1$ chosen cells,
\begin{align*}
V_i(S_j)
&\le V_i(D_j)+\sum_{r=n-\ell(j)-1}^{n-1}V_i(Q_{j,r})\\
&\le V_i(D_j)+\sum_{r=n-\ell(j)-1}^{n-1}V_i(Q_{i,r-1})
\tag{Round robin}\\
&= V_i(D_j)+\sum_{r=n-\ell(j)-2}^{n-2}V_i(Q_{i,r})
\tag{Re-indexing}\\
&\le V_i(D_j)+\sum_{r=n-\ell(i)}^{n-2}V_i(Q_{i,r}).
\tag{$\ell(i)\ge\ell(j)+2$}
\end{align*}
When agent $i$ picked her last cell, the available cells were together worth at least $\frac12 V_i(R)$ to her, and there were fewer than $1/(2\eps)$ of them, as shown above. Hence $V_i(Q_{i,n-1})>\eps V_i(R)$. Therefore, $V_i(D_i)+V_i(Q_{i,n-1}) > V_i(D_i)+\eps V_i(R)$, which is at least $V_i(D_j)$, by the hypothesis of~\Cref{lem:repair main}. Combining this with the preceding bound on $V_i(S_j)$ we have $V_i(S_j)< V_i(D_i)+\sum_{r=n-\ell(i)}^{n-1}V_i(Q_{i,r}) \le V_i(S_i)$, where the last inequality uses the fact that $S_i$ contains $D_i$ and the last $\ell(i)$ chosen cells. This proves~\eqref{eq:levels}.

\paragraph{No old bundle is cut.}
Suppose, towards a contradiction, that some old bundle is cut, so some piece has level $-1$. By property (i) of~\Cref{lem:permute}, some piece is allocated whole and therefore has level $n-1$. Since there are $n$ pieces and $n+1$ possible levels, some $\lambda\in\{0,\dots,n-2\}$ is not the level of any piece. Let $H=\{i\in N:\ell(i)>\lambda\}$ and $L=\{i\in N:\ell(i)<\lambda\}$ be the sets of agents whose corresponding pieces $S_i$ have levels above and below $\lambda$, respectively. These sets are nonempty and partition $N$. Furthermore, for all $i\in H$ and $j\in L$, $\ell(i)\ge\ell(j)+2$, and therefore, by Inequality~\eqref{eq:levels}, $V_i(S_i)>V_i(S_j)$. Since $S_i$ is held by some agent, by the envy-freeness of $W$ (\Cref{lem:permute}) we have that $V_i(W_i)\ge V_i(S_i)$; so, overall, for all $i\in H$ and $j\in L$, $V_i(W_i) > V_i(S_j)$.

Thus, no agent $i\in H$ can receive a piece $S_j$ with $j \in L$, since $W_i=S_j$ would contradict this strict inequality. Therefore, every agent $i\in H$ receives a piece $S_j$ with $j\in H$. There are exactly $|H|$ such pieces, so the agents in $H$ hold all of them, and agents in $L$ hold pieces whose level is below $\lambda$. 
Each piece $S_k$ held by an agent in $L$ is a proper suffix of $P_k$: otherwise $S_k=P_k$, which would give $\ell(k)=n-1$, whereas $\ell(k)<\lambda\le n-2$. We can therefore apply property (ii) of~\Cref{lem:permute}: there exist agents $i\in H$ and $j\in L$ such that $V_i(W_j)=V_i(W_i)$. Since agent $j$ belongs to $L$, she receives $W_j=S_k$ for some $k\in L$. However, since $i\in H$ and $k\in L$, $V_i(W_i)>V_i(S_k)=V_i(W_j)$, contradicting $V_i(W_j)=V_i(W_i)$.
\end{proof}

\section{Implementing \Permute}\label{sec:permute}

In this section, we implement \Permute{} and prove~\Cref{lem:permute}. \Permute{} takes as input a set $N$ of $n$ agents and $n$ pairwise disjoint pieces $P_1,\dots,P_n$ (some possibly empty), and outputs a partial envy-free allocation $Y$ with residue $R$, in which every agent receives a suffix of a different input piece, with the additional guarantees stated in~\Cref{lem:permute}.

The protocol tries every ordering $i_1,\dots,i_n$ of the agents and every ordering $P_{j_1},\dots,P_{j_n}$ of the input pieces. For each guess, we process the agents in that order. Agent $i_r$ starts with piece $P_{j_r}$ whole; whenever an earlier agent envies it, we remove just enough cake from its ``left'' to eliminate that envy, keeping the remaining suffix. When all agents have been processed, we test whether the resulting allocation is envy-free. If so, we return it; otherwise, we try the next guess. 

\begin{algorithm}[H]
\caption{$\Permute(N,(P_i)_{i \in N})$, where $n=|N|$}\label{alg:permute}
\ForEach{ordering $i_1,\dots,i_n$ of the agents and ordering $P_{j_1},\dots,P_{j_n}$ of the input pieces}{
  \For{$r=1,\dots,n$}{
    $Y_{i_r}\gets P_{j_r}$\;
    \For{$s=1,\dots,r-1$}{
      \If{$V_{i_s}(Y_{i_r})>V_{i_s}(Y_{i_s})$}{
        Remove from $Y_{i_r}$ its smallest prefix worth $V_{i_s}(Y_{i_r})-V_{i_s}(Y_{i_s})$ to agent $i_s$\;\label{permute:trim}
      }
    }
  }
  \If{$Y$ is envy-free}{
    \Return{$Y$ and its residue $\left(\bigcup_{i \in N} P_i\right)\setminus\left(\bigcup_{i\in N}Y_i\right)$}\;\label{permute:return}
  }
}
\end{algorithm}

We call the allocation constructed for a fixed pair of orders, before the envy-freeness test, its \emph{candidate}. A candidate piece keeps as much cake as possible, including any zero-valued parts that need not be removed. Therefore, the candidate piece $Y_{i_r}$ contains every suffix $T$ of $P_{j_r}$ satisfying
\begin{equation}\label{eq:permute-caps}
V_{i_s}(T)\le V_{i_s}(Y_{i_s})
\qquad\text{for every }s<r.
\end{equation}
Indeed, $Y_{i_r}$ starts as the whole piece $P_{j_r}$. Whenever an earlier agent causes a trim, removing the smallest prefix preserves every suffix satisfying that agent's inequality. Thus every suffix satisfying all the constraints in~\eqref{eq:permute-caps} survives all the trims. We restate and prove~\Cref{lem:permute}.

\permutelemma*

\begin{proof}[Proof of~\Cref{lem:permute}]
For a successful guess, the bijection in~\Cref{lem:permute} is given by $\pi(i_r)=j_r$ for $r=1,\dots,n$. As long as we can show that some guess succeeds, we immediately have that every returned allocation is envy-free. Furthermore, each agent $i_r$ receives a suffix of a different input piece $P_{j_r}$, and the first agent $i_1$ receives $P_{j_1}$ whole, proving property~(i). For property~(ii), let $S$ be a nonempty set of agents, each of whom receives a proper suffix of her assigned input piece, and let $i_r$ be the first agent in $S$ in the guessed order. Since $Y_{i_r}\ne P_{j_r}$, some earlier agent caused a trim of $Y_{i_r}$; let $i_s$ be the last such agent. That trim makes $V_{i_s}(Y_{i_r})=V_{i_s}(Y_{i_s})$, and neither piece changes afterward. Since $s<r$ and $i_r$ is the first agent in $S$, we have $i_s\notin S$, as required.
Finally, for the query count, there are $(n!)^2$ guesses and each guess makes at most one prefix cut for each pair of positions $s<r$ (so at most $\binom n2$ prefix cuts). All other operations use recorded values (see~\Cref{sec:conventions}). 

It remains to prove that some guess succeeds. First, we show that, for any fixed guess of orders, the corresponding candidate ``contains'' every envy-free allocation in which each agent receives a suffix of her assigned input piece. Second, for every $\delta>0$, we construct a $\delta$-envy-free allocation which ``contains'' some candidate. Finally, we show that, if every candidate failed, every allocation produced by this construction would have envy at least some fixed positive amount (independent of $\delta$). Choosing $\delta$ smaller than this amount gives a contradiction.

\paragraph{Every envy-free suffix allocation lies inside its candidate.}
We prove the following claim.
\begin{claim}\label{claim:containment}
Fix an agent order $i_1,\dots,i_n$, an input-piece order $P_{j_1},\dots,P_{j_n}$, and let $Y$ be the corresponding candidate. Suppose that $Z$ is some envy-free allocation such that every $Z_{i_r}$ is a suffix of $P_{j_r}$. Then $Z_{i_r}\subseteq Y_{i_r}$ for all $r = 1, \dots, n$.
\end{claim}

\begin{proof}
We will prove the statement by induction on $r$. In $Y$, the first agent receives her input piece whole, so the claim holds for $r=1$. For a later position $r$, for every $s<r$, envy-freeness of $Z$ and the induction hypothesis give $V_{i_s}(Z_{i_r}) \le V_{i_s}(Z_{i_s}) \le V_{i_s}(Y_{i_s})$. Thus $Z_{i_r}$ satisfies all the constraints in~\eqref{eq:permute-caps}, and therefore $Z_{i_r}\subseteq Y_{i_r}$.
\end{proof}

\paragraph{Creating a $\delta$-envy-free suffix allocation.}
Fix some $\delta>0$, and partition each input piece $P_i$ into finitely many consecutive cells, each worth at most $\delta$ to every agent. Such cells exist: e.g., as in~\Repair{}, we can ask each agent to mark the points that split each piece into $\lceil1/\delta\rceil$ parts of equal value to her, and cut at all the resulting marks (observing that each $P_i$ is worth at most $1$ to each agent; a piece worth zero to an agent needs no marks from that agent). Note that this construction and the process below are used only in the proof; they are not steps of \Permute{} and do not contribute to its query count.

Match the agents to the input pieces arbitrarily, and let $X$ be the allocation that gives each agent the empty suffix of her input piece, i.e., $X_i=\emptyset$ for every agent $i$; throughout, each suffix keeps the label of its input piece, even when it is empty. We run a version of the envy-cycle elimination algorithm of Lipton et al.~\cite{lipton2004approximately}, in which the ``items'' are the cells, and the only item that can be given to an agent is the cell immediately preceding the suffix she currently holds. The envy graph of $X$ has an edge $i\to j$ whenever agent $i$ strictly envies agent $j$, i.e., whenever $V_i(X_i)<V_i(X_j)$. Repeat the following:
\begin{enumerate}[leftmargin=*]
\item If the envy graph has a directed cycle, give each agent on the cycle the suffix held by the next agent on the cycle (labels follow suffixes).
\item Otherwise, if some agent that nobody strictly envies holds a suffix that is not the whole input piece from which her current suffix comes, add to her suffix the cell immediately preceding it.
\item Otherwise, stop.
\end{enumerate}

First, observe that this process terminates. Each enlargement allocates a cell that was not allocated before, and there are finitely many cells, so there are finitely many enlargements. During any sequence of consecutive cycle rotations, the pieces do not change, so $\sum_{i\in N}V_i(X_i)$ takes at most $n!$ values. Every rotation strictly increases this sum, so every such sequence is finite. At termination, the envy graph is acyclic (otherwise step (1) would apply), and every agent holding a non-whole suffix is envied by someone (otherwise step (2) would apply).
Second, throughout this process, $X$ is $\delta$-envy-free, i.e., $V_i(X_i)\ge V_i(X_j)-\delta$ for all $i, j \in N$. This holds initially because all allocated pieces are empty. A cycle rotation strictly improves each participating agent's own value and leaves the collection of pieces unchanged, so it preserves $\delta$-envy-freeness. An enlargement adds a cell worth at most $\delta$ to every agent to a piece that nobody envied, so it also preserves $\delta$-envy-freeness.

\paragraph{$X$ contains a candidate.}
Let $X$ be the allocation at termination. Since its envy graph is acyclic, order the agents in a topological order $i_1,\dots,i_n$ (i.e., every envy-edge goes from an earlier agent to a later agent). For each position $r$, let $j_r$ be the label of the input piece whose suffix agent $i_r$ holds. These labels give an ordering $P_{j_1},\dots,P_{j_n}$ of the input pieces. Let $Y$ be the candidate produced by \Permute{} for these two orders. We will show, by induction on $r$, that $Y_{i_r}\subseteq X_{i_r}$ for all $r = 1, \dots, n$.

For $r=1$, agent $i_1$ has no incoming edge in the envy graph, so nobody envies her. Since the process has terminated, she must hold a whole input piece (otherwise step~(2) would apply), i.e., $X_{i_1}=P_{j_1}$. The candidate $Y$ also gives the first agent her input piece whole, so $Y_{i_1}=P_{j_1}=X_{i_1}$.
Now let $r>1$ and suppose that $Y_{i_s}\subseteq X_{i_s}$ for every $s<r$.  If $X_{i_r}=P_{j_r}$, the claim is immediate. Otherwise, $i_r$ holds a non-whole suffix in $X$, and therefore some agent $i_s$ envies $X_{i_r}$. By our choice of the agent order, $s<r$ and the induction hypothesis gives $Y_{i_s}\subseteq X_{i_s}$. Therefore,
\[
V_{i_s}(Y_{i_r})
\le V_{i_s}(Y_{i_s})
\le V_{i_s}(X_{i_s})
< V_{i_s}(X_{i_r}),
\]
where the first inequality holds because an earlier agent never envies a later agent in a candidate $Y$. Since $X_{i_r}$ and $Y_{i_r}$ are suffixes of the same input piece $P_{j_r}$, one contains the other; the strict inequality rules out $X_{i_r}\subseteq Y_{i_r}$. Therefore $Y_{i_r}\subsetneq X_{i_r}$, completing the induction.

\paragraph{Some candidate is envy-free.}
Fix a guess with orders $i_1,\dots,i_n$ and $P_{j_1},\dots,P_{j_n}$, and let $Y$ be the corresponding candidate.  An \emph{extension} of $Y$ is an allocation $X$ in which every $X_{i_r}$ is a suffix of $P_{j_r}$ containing $Y_{i_r}$. Suppose that candidate $Y$ is not envy-free.

First, no extension of $Y$ is envy-free either: by~\Cref{claim:containment}, an envy-free extension $X$ would satisfy $X_i\subseteq Y_i$ for every agent $i$, and hence $X=Y$, a contradiction. Second, as $X$ ranges over the extensions of $Y$, the length (the sum of the lengths of its component intervals) of each suffix $X_{i_r}$ ranges from the length of $Y_{i_r}$ to that of $P_{j_r}$. Thus, the possible vectors of suffix lengths form a closed bounded box in $\mathbb R^n$. Let $e(X)=\max_{i,j\in N}\bigl(V_i(X_j)-V_i(X_i)\bigr)$ be the maximum envy in an allocation $X$. By nonatomicity, every $V_i(X_j)$ is continuous in the length of $X_j$, and therefore $e(X)$ is continuous in these suffix lengths. A continuous function on a closed, bounded box attains its minimum. Since no extension is envy-free, this minimum is strictly positive:
\[
\gamma_Y=\min_{X\text{ extending }Y}e(X)>0.
\]
Therefore, every extension of the failed candidate $Y$ has envy at least $\gamma_Y$.

Finally, suppose, towards a contradiction, that every candidate fails. There are finitely many guesses, so $\gamma=\min_Y\gamma_Y>0$, where the minimum ranges over all candidates. Apply the preceding construction with $\delta=\gamma/2$. It produces a $\delta$-envy-free allocation $X$ containing the candidate $Y$ of some guess, with the same assignment of input pieces. Thus $X$ is an extension of $Y$, but $e(X)\le\delta<\gamma\le\gamma_Y$,
a contradiction. Therefore some guess succeeds, completing the proof.
\end{proof}

\bibliographystyle{alpha}
\bibliography{refs}

\end{document}